\documentclass[11pt,reqno]{article}

\usepackage[leqno]{amsmath}
\usepackage{amsthm}
\usepackage{amssymb}
\usepackage{times}
\usepackage{verbatim}
\usepackage{enumerate}
\usepackage{geometry}
\usepackage[pdftex]{graphicx}
\usepackage{xspace}
\allowdisplaybreaks[1] 

\newcommand\myeq{\mathrel{\overset{\makebox[0pt]{\mbox{\normalfont\tiny\sffamily def}}}{=}}}

\newtheorem{theorem}{Theorem}[section]
\newtheorem{conjecture}{Conjecture}[section]
\newtheorem{lemma}[theorem]{Lemma}
\newtheorem{corollary}[theorem]{Corollary}

\newtheorem{definition}{{\sc Definition}\rm}[section]

\newcommand{\SLIDING}{\textsc{Sliding-Window}\xspace}

\newcommand{\E}{\mathbb{E}}

\newcommand{\ignore}[1]{\relax}

\begin{document}

\title {Sliding Windows with Limited Storage}
\author{Paul Beame\\
  {\small Computer Science and Engineering}\\
  {\small University of Washington}\\
  {\small Seattle, WA 98195-2350}\\
  {\small\tt beame@cs.washington.edu}
  \and
  Rapha\"el Clifford\\
  {\small Department of Computer Science}\\
  {\small University of Bristol}\\
  {\small Bristol BS8 1UB, United Kingdom}\\
  {\small\tt clifford@cs.bris.ac.uk }\\
  \and
  Widad Machmouchi\\
  {\small Computer Science and Engineering}\\
  {\small University of Washington}\\
  {\small Seattle, WA 98195-2350}\\
  {\small\tt widad@cs.washington.edu}
  }

\date{April 2, 2013
  \vspace*{-3ex}}

\maketitle
\thispagestyle{empty}
\begin{abstract}
\noindent The results of this paper are superceded by the paper at: \verb|http://arxiv.org/abs/1309.3690|.
\vfill

We consider time-space tradeoffs for exactly computing frequency
moments and order statistics over sliding windows~\cite{dgim:windows}.
Given an input of length $2n-1$, the task is to output the function of
each window of length $n$, giving $n$ outputs in total.
Computations over sliding windows are related to direct sum problems
except that inputs to instances almost completely overlap.
\begin{itemize}
\item We show an average case and randomized time-space tradeoff lower bound of
$T\cdot S \in \Omega(n^2)$ for multi-way branching programs, and
hence standard RAM and word-RAM models, to compute the number
of distinct elements, $F_0$, in sliding windows over alphabet $[n]$.
The same lower bound holds for computing the low-order bit of $F_0$ and
computing any frequency moment $F_k$ for $k\ne 1$.
We complement this lower bound with a $T\cdot S \in \tilde O(n^2)$
deterministic RAM algorithm for exactly computing $F_k$ in sliding windows.
\item We show time-space separations between the complexity of sliding-window element distinctness and that of sliding-window $F_0\bmod 2$
computation.
In particular for alphabet $[n]$ there is a very simple errorless
sliding-window algorithm for element distinctness that runs in $O(n)$ time on
average and uses $O(\log{n})$ space.
\item We show that any algorithm for a single element distinctness instance
can be extended to an algorithm for the sliding-window version of element
distinctness with at most a polylogarithmic increase in the time-space product.
\item Finally, we show that the sliding-window computation of
order statistics such as the maximum and minimum can be computed with only a
logarithmic increase in time, but that a $T\cdot S \in \Omega(n^2)$ lower
bound holds for sliding-window computation of order statistics such as the
median, a nearly linear increase in time when space is small.
\end{itemize}
\end{abstract}

\newpage
\setcounter{page}{1}
\section{Introduction}

Direct sum questions in a computational model ask how the complexity of
computing many instances of a function $f$ on independent inputs increases
as the number of instances grows.  The ideal direct sum theorem shows that
computing $n$ independent instances of $f$ requires an $\Omega(n)$ factor
increase in computing resources over computing a single instance of $f$.

Valuable though direct sum theorems can be, they require an increase
in the number of inputs equal to the number of instances.
We are interested in how the complexity of computing many copies of a
function $f$ can grow when the inputs overlap so that the total size of the
input is not much larger than the input size for a single
function\footnote{Computing many copies of a function on overlapping inputs
(selected via a combinatorial design) was used as the basis for the
Nisan-Wigderson pseudorandom generator construction~\cite{nw:pseudorandom},
though in that case the total
input size is much larger than that of the original function.}.

A particularly natural circumstance in which one would want to evaluate many
instances of a function on overlapping inputs occurs in the context of time
series analysis.
For many functions computed over sequences of data elements or data
updates, it is useful to know the value of the function on many
different intervals or {\em windows} within the sequence, each representing
the recent history of the data at a given instant.
In the case that an answer
for every new element of the sequence is required, such computations have
been termed {\em sliding-window} computations for the associated
functions~\cite{dgim:windows}.

We focus on the questions of when the sliding-window versions of problems
increase their complexity, and under what circumstances one can prove
significantly larger lower bounds for these sliding-window versions than
can be shown for the original functions.
Unlike an ordinary direct sum lower bound, a positive answer will yield a
proportionately better lower bound relative to the input size.
The complexity measure we use is the time required for a given amount of
storage; i.e., we study time-space tradeoffs of these sliding-window problems.
Given the general difficulty of proving lower bounds for single output
functions, in addition to the goal of obtaining proportionately larger lower
bounds, comparing the difficulty of computing sliding-window versions
of functions $f$ and $g$ may be easier than comparing them directly.

Many natural functions have previously been studied for sliding windows
including entropy, finding frequent symbols, frequency moments and order
statistics, which can be computed approximately in small
space using randomization even in one-pass data stream
algorithms~\cite{dgim:windows,BabcockBDMW02,am04,lt06, lt06b, ccm07,boz12}.
Approximation is required since exactly computing these values in this online
model can easily be shown to require large space.
The interested reader may find a more comprehensive list of sliding-windows
results by following the references in~\cite{boz12}.

We focus on many of these same statistical functions and consider them
over inputs of length $2n-1$ where
the sliding-window task is to compute the function
for each window of length $n$, giving $n$ outputs in total.
We write $f^{\boxplus n}$ to denote this sliding-window version of a function
$f$.

Our main results concern the computation of frequency moments and element
distinctness over sliding windows.
Frequency moment $F_0$ is the number of distinct elements in the input.
The element distinctness problem $ED$, determining whether all input elements
are distinct, is the special case of testing whether $F_0$ is equal to the number
of inputs.   $ED$ is often considered the decision problem that best
tracks the complexity of integer sorting,
a problem for which we already know tight time-space tradeoff
lower bounds~\cite{bc:sorting,bea:sorting} on general sequential computation
models like multi-way branching programs and RAMs, as well as matching
comparison-based upper bounds~\cite{pr:comparison-sorting}.
(As is usual, the input is
assumed to be stored in read-only memory and the output in write-only memory
and neither is counted towards the space used by any algorithm.
The multi-way branching program model simulates standard RAM models that
are unit-cost with respect to time and log-cost with respect to space.
Therefore in discussing complexity, we measure space usage in bits
rather than words.)

We prove time-space lower bounds for computing the sliding-window version
of any frequency moment $F_k$ for $k\ne 1$.  In particular, the time $T$ and
space $S$ to compute $F_k^{\boxplus n}$ must satisfy
$T\cdot S \in \Omega(n^2)$.
($F_1$ is simply the size of the input, so computing its value is always
trivial.)
Moreover, we show that the same lower bound holds for computing just
the parity of the number of distinct elements, $F_0\bmod 2$, in each window.
The bounds are proved directly for multi-way branching programs which imply
lower bounds for the standard RAM and word-RAM models, as well as for the data
stream models discussed above.
The best previous lower bounds for computing any of these sliding window
problems are much smaller time-space tradeoff lower bounds that apply to the
computation of a single instance of $F_k$.
In particular, for any $k\ne 1$, an input has distinct elements if any only if
$F_k=n$, so these follow from previous lower bounds for $ED$.
Ajtai~\cite{ajtai:nonlinear-journal}
showed that any linear time solution for $ED$ (and hence $F_k$)
must use linear space.   No larger time lower bound is known unless the space
$S$ is $n^{o(1)}$.  In that case, the best previous lower bound for computing
$ED$ (and hence $F_k$) is a $T \in \Omega(n\sqrt{\log (n/S)/\log\log(n/S)})$
lower bound shown in~\cite{bssv:randomts-journal}.
This is substantially smaller than our $T\cdot S\in \Omega(n^2)$ lower bound.

We complement our lower bound with a comparison-based RAM algorithm for
any $F_k^{\boxplus n}$ which has $T\cdot S\in \tilde O(n^2)$, showing that
this is nearly an asymptotically tight bound, since it provides a general RAM
algorithm that runs in the same time complexity for any polynomial-sized input
alphabet\footnote{As is usual, we use $\tilde O$ to suppress polylogarithmic
factors in $n$.}.

Our lower bounds for frequency moment computation hold for randomized
algorithms even with small success probability $2^{-O(S)}$ and for the average
time and space used by deterministic algorithms on inputs in which the values
are independently and uniformly chosen from $[n]$.

It is interesting to contrast our lower bounds for the sliding-window
version of $F_0\bmod 2$ with those for the sliding-window version of
$ED$.
It is not hard to show that on average for integers independently and uniformly
chosen from $[n]$, $ED$
can be solved with $\overline T \cdot \overline S \in \tilde O(n)$.  This
can be extended to an algorithm that has a similar
$\overline T \cdot \overline S \in \tilde O(n)$ bound for $ED^{\boxplus n}$
on this input distribution.
This formally proves a separation between the complexity of sliding-window
$F_0\bmod 2$ and sliding-window $ED$.
Interestingly, this separation is not known to exist for one window alone.

In fact, we show that the similarity between the complexities of computing $ED$
and $ED^{\boxplus n}$ on average also applies to the worst-case complexity
of deterministic and randomized algorithms.
We give a general reduction which shows that for
any space bound $S$, by using space $S^*\in S+O(\log^2 n)$, one can convert any
algorithm $A$ for $ED$ running in time $T$ into an algorithm $A^*$ that solves
$ED^{\boxplus n}$ in time $T^*\in O(T\log^2 n)$ or alternatively $T^*\in O(T\log n)$
if $T\in \Omega(n^{1+\delta})$.  That is, there is no sliding-window analogue
of a direct sum result for $ED$.

These results suggest that in the continuing search for strong lower complexity
lower bounds, $F_0\bmod 2$ may be a better choice as a difficult decision
problem than $ED$.

Finally, we discuss the problem of computing the $t^{th}$ order statistic in
each window.  For these problems we see the full range of relationships between
the complexities of the original and sliding-window versions of the problems.
In the case of $t=n$ (maximum) or $t=1$ (minimum) we show that
computing these properties over sliding windows can be solved by a comparison
based algorithm in $O(n\log n)$
time and only $O(\log n)$ bits of space and hence there is no sliding-windows
analogue of a direct sum result for these problems.
In contrast, we show that a
$T\cdot S \in \Omega(n^2)$ lower bound holds
when $t  = \alpha n$ for any fixed $0 < \alpha < 1$.
Even for algorithms that only use comparisons, the expected time for
errorless randomized algorithms to find the median in a single window is 
$\overline T\in \Omega(n\log\log_S n)$~\cite{chan:selection-journal} and
there is an errorless randomized algorithm that precisely matches this
bound~\cite{chan:selection-journal}.
Hence for many values of $S$ there is an approximate direct
sum analogue for these sliding-window order statistics.

\paragraph{Related work}

While sliding-windows versions of problems have been considered in the context
of online and approximate computation, there is little research that has
explicitly considered any such problems in the case of exact offline
computation.
One instance where a sliding-windows problem has been considered is
a lower bound for generalized string
matching due to Abrahamson~\cite{abr87}.
This lower bound implies that for any fixed string $y\in [n]^n$
with $n$ distinct values, $H_y^{\boxplus n}$ requires
$T\cdot S\in \Omega(n^2/\log n)$ where decision problem $H_y(x)$ is 1 if and
only if the Hamming distance between $x$ and $y$ is $n$.
This bound is an $\Omega(\log n)$
factor smaller than our lower bound for sliding-window $F_0\bmod 2$.

One of the main techniques to derive time-space tradeoffs for
branching programs was introduced in \cite{bc:sorting} by
Borodin and Cook and was generalized to a number of other problems (e.g.,
\cite{yes84,abr87,abr:tradeoff,bea:sorting,mnt:hashing-journal,sw:multiplyRAM}).
Our lower bounds draw on this method but require some additional
work to adapt it to the case of computing frequency moments.

In addition to lower bounds that apply to unrestricted models such as RAMs
and general branching program models, some of the problems we consider have
been considered in structured comparison-based models.
Borodin et al.~\cite{bfmuw87} gave a time-space tradeoff lower bound for
computing $ED$ (and hence any $F_k$ for $k\ne 1$) on comparison branching
programs of $T^2\cdot S\in\Omega(n^3)$ and since $S\geq \log_2{n}$, $T\cdot S\in\Omega(n^{3/2}\sqrt{\log n})$).
Yao~\cite{yao88} improved this to a near-optimal
$T\cdot S \in \Omega(n^{2-\epsilon(n)})$, where $\epsilon(n)= 5/(\ln n)^{1/2}$.
Since our algorithm for computing $F_k^{\boxplus n}$ is comparison-based,
this lower bound is not far from matching our upper bound for the
sliding-window version of $F_k$.

Finally, we note that previous research implies a separation between the
complexities of $ED$ and $F_0\bmod 2$ in the context of quantum query
algorithms: $ED$ has quantum query complexity $\Theta(n^{2/3})$ (lower bound
in $\cite{as:quantum-collision}$ and matching quantum query algorithm
in~\cite{ambainis:distinctness}).  On other hand, the 
lower bounds in~\cite{bm:qac0} imply that $F_0\bmod 2$ has quantum query
complexity $\Omega(n)$.

\paragraph{Organization}
In the remainder of this section we more formally define
the $\boxplus$ operator, the statistical functions we consider, and the
multi-way branching program model.
In Section~\ref{sec:fk} we present our lower bound for computing frequency
moments $F_k$ and $F_0\bmod 2$ over sliding windows followed by a
comparison-based algorithm that yields a nearly matching upper bound.
In Section~\ref{sec:ed} we give our algorithms for element distinctness over
sliding windows which show the separation between $F_0\bmod 2$ and $ED$.
Finally in Section~\ref{sec:order} we give our upper and lower bounds for
sliding-window computation of order statistics.

\paragraph{Sliding Windows} Let $D$ and $R$ be two finite sets and
$f:D^n\rightarrow R$ be a function over strings of length $n$. We
define the operation \SLIDING, denoted $\boxplus$, that
takes $f$ and returns a function $f^{\boxplus t}:D^{n+t-1}\rightarrow R^t$,
defined by
$f^{\boxplus t} (x) = \left(f(x_i\ldots x_{i+n-1})\right)_{i=1}^{t}$.
We concentrate on the case that $t=n$ and
apply the \SLIDING operator to the
functions $F_k$, $F_k \bmod 2 $, $ED$, and $O_t$, the $t^{th}$ order statistic.
We will use the notation $F_k^{(j)}$
(resp. $f_i^{(j)}$) to denote the $k^{th}$ frequency moment (resp. the frequency of symbol $i$)
of the string in the window of length $n$ starting at position $j$.

\paragraph{Frequency Moments, Element Distinctness, and Order Statistics}
Let $a =  a_1a_2\ldots,a_n$ be a string of $n$ symbols from a
linearly ordered set.
We define the \textit{$k^{th}$ frequency
moment} of $a$, $F_k(a)$, as $F_k(a)= \sum_{i\in D} f_i^k$, where $f_i$ is the
frequency (number of occurrences) of symbol $i$ in the string $a$ and $D$ is
the set of symbols that occur in $a$.
Therefore, $F_0(a)$ is the number of distinct
symbols in $a$ and $F_1(a) = |a|$ for every string $a$.
The \textit{element distinctness} problem is a decision problem
defined as: $ED(a) =1  \mbox{ if } F_0(a) = |a| \mbox{ and } 0
\mbox{ otherwise}.$
We write $ED_n$ for the $ED$ function restricted to inputs $a$ with $|a|=n$.
The \textit{$t^{th}$ order statistic} of $a$,
$O_t$, is the $t^{th}$ smallest symbol in $a$.
Therefore $O_n$ is
the maximum of the symbols of $a$ and $O_{\lceil \frac{n}{2}\rceil}$
is the median.

\paragraph{Branching programs}
Let $D$ and $R$ be finite sets and $n$ and $m$ be two positive integers.
A \textit{$D$-way branching program} is a
connected directed acyclic graph with special nodes: the
\textit{source node} and possibly many \textit{sink nodes}, a set of
$n$ inputs and $m$ outputs.
Each non-sink node is labeled with an input index and every edge is
labeled with a symbol from $D$, which corresponds to the value of
the input indexed at the originating node.
In order not to count the space required for outputs, as is standard,
we assume that each edge can be labelled by some set of output assignments.
For a directed path $\pi$ in a branching program, we call the set of indices of
symbols queried by $\pi$ the \emph{queries} of $\pi$, denoted by $Q_{\pi}$;
we denote the \emph{answers} to those queries by $A_{\pi}:Q_\pi \rightarrow D$
and the outputs produced along $\pi$ as a \emph{partial}
function $Z_\pi:[m]\rightarrow R$.

A branching program computes a function $f:D^n\rightarrow R^m$ by
starting at the source and then proceeding along the nodes of the
graph by querying the inputs associated with each node and following
the corresponding edges.
In the special case that there is precisely one output, without loss of
generality, any edge with this output may instead be assumed to be unlabelled
and lead to a unique sink node associated with its output value.

A branching program $B$ is said to \textit{compute} a function $f$
if for every $x \in D^n$, the output of $B$ on $x$, denoted $B(x)$,
is equal to $f(x)$.
A \textit{computation} (in $B$) on $x$ is a
directed path, denoted $\pi_B(x)$, from the source to a sink in $B$ whose
queries to the input are consistent with $x$.
The time $T$ of a branching program is the length of the
longest path from the source to a sink and the space $S$ is the
logarithm base 2 of the number of the nodes in the branching
program. Therefore, $S \geq \log T$ where we write $\log x$ to denote
$\log_2 x$.

A branching program $B$ \textit{computes $f$ under
$\mu$ with error at most $\eta$} iff $B(x)= f(x)$ for all but an
$\eta$-measure of $x \in D^n$ under distribution $\mu$.
A \textit{randomized} branching program
$\mathcal{B}$ is a probability distribution over deterministic
branching programs with the same input set. $\mathcal{B}$
computes a function $f$ with error at most $\eta$ if for every
input $x \in D^n$, $\Pr_{B\sim\mathcal{B}}[B(x) = f(x)] \geq
1-\eta$. The time (resp. space) of a randomized branching program
is the maximum time (resp. space) of a deterministic branching program
in the support of the distribution.

A branching program is \textit{levelled} if  the nodes are divided
into an ordered collection of sets each called a \textit{ level}
where edges are between consecutive levels only.
Any branching program can be leveled by increasing the space
$S$ by an additive factor of $\log T$. Since $S \geq \log T$, in
the following we assume that our branching programs are
leveled.


\section{Frequency Moments over Sliding Windows}\label{sec:fk}

We begin with our main lower bound for computing frequency moments over
sliding windows and then derive a nearly matching upper bound.

\subsection{A general sequential lower bound for $F_k^{\boxplus n}$ and $(F_0\bmod 2)^{\boxplus n}$}

We derive a time-space tradeoff lower bound for randomized
branching programs computing $F_k^{\boxplus n}$ for $k=0$ and $k\ge 2$.
Further, we show that the lower bound also holds for computing 
$(F_0\bmod 2)^{\boxplus n}$. (Note that the parity of $F_k$ for $k\ge 1$ is
exactly equal to the parity of $n$; thus the outputs of $(F_k\bmod 2)^{\boxplus n}$ are all equal to $n\bmod 2$.)

\begin{theorem}\label{main_theorem}
Let $k=0$ or $k\ge 2$. There is a constant $\delta>0$ such that any $[n]$-way
branching program of time $T$ and space
$S$ that computes $F_k^{\boxplus n}$ with error at most $\eta$,
$0<\eta<1-2^{-\delta S}$, for input randomly chosen uniformly from $[n]^{2n-1}$ must
have $T\cdot S \in \Omega(n^2)$. The same lower bound holds for $(F_0\bmod 2)^{\boxplus n}$.
\end{theorem}

\begin{corollary}
\label{average-and-random}
Let $k=0$ or $k\ge 2$.
\begin{itemize}
\item The average time $\overline T$ and average space
$\overline S$ needed
to compute $(F_k)^{\boxplus n}(x)$ for $x$ randomly chosen uniformly
from $[n]^{2n-1}$ satisfies $\overline T\cdot \overline S\in \Omega(n^2)$.
\item For $0<\eta<1-2^{-\delta S}$, any $\eta$-error randomized RAM or
word-RAM algorithm
computing $(F_k)^{\boxplus n}$ using time $T$ and space $S$
satisfies $T\cdot S\in\Omega(n^2)$.
\end{itemize}
\end{corollary}

\begin{proof}[Proof of Theorem~\ref{main_theorem}]
We derive the lower bound for $F_0^{\boxplus n}$ first.
Afterwards we show the modifications
needed for $k \geq 2$ and for computing
$(F_0\bmod 2)^{\boxplus n}$.
For convenience, on input $x\in [n]^{2n-1}$, we write $y_i$ for the output
$F_k(x_i,\ldots,x_{i+n-1})$.

We use the general approach
of Borodin and Cook~\cite{bc:sorting} together with the observation
of~\cite{abr:tradeoff} of how it applies to average case complexity and
randomized branching programs.
In particular, we divide the branching program $B$ of length $T$ into layers
of height $q$ each.
Each layer is now a collection of small branching programs $B'$, each of
whose start node is a node at the top level of the layer.
Since the branching program must produce $n$ outputs for each input $x$,
for every input $x$ there exists
a small branching program $B'$ of height $q$ in some layer that produces at
least $nq/T>S$ outputs.
There are at most $2^S$ nodes in $B$ and hence there are at most $2^S$ such
small branching programs among all the layers of $B$.
One would normally prove that the fraction of $x\in [n]^{2n-1}$ for which
any one such small program correctly produces the desired number of outputs
is much smaller than $2^{-S}$ and hence derive the desired lower bound.
Usually this is done by arguing that the fraction of inputs consistent with
any path in such a small branching program for which a fixed set of outputs
is correct is much smaller than $2^{-S}$.

This basic outline is more complicated in our argument.
One issue is that if a path in a small program $B'$ finds that certain
values are equal, then the answers to nearby windows may be strongly correlated
with each other; for example, if $x_i=x_{i+n}$ then $y_i=y_{i+1}$.
Such correlations risk making the likelihood too high that the correct
outputs are produced on a path. Therefore, 
instead of considering the total number of outputs produced, 
we reason about the number of outputs from positions that are not duplicated
in the input and argue that with high probability there will be a linear number
of such positions.

A second issue is that inputs for which
the value of $F_0$ in a window happens to be extreme, say $n$ - all distinct -
or 1 - all identical, allow an almost-certain prediction of the value of
$F_0$ for the next window.
We will use the fact that under the uniform distribution, cases like these
almost surely do not happen; indeed the numbers of
distinct elements in every window almost surely fall in a range close to
their mean and in this case the value in the next window will be
predictable with probability bounded below 1 given the value in the previous
ones.
In this case we use the chain rule to compute the overall probability
of correctness of the outputs.

We start by analyzing the likelihood that an output of $F_0$ is extreme.

\begin{lemma}\label{lemma:f0_range_prob_a}
Let $a$ be chosen uniformly at random from $[n]^{n}$.
Then the probability that $F_0(a)$ is
between $0.5n$ and $0.85n$ is at least $1-2 e^{-n/50}$.
\end{lemma}

\begin{proof}
For $a=a_1\ldots a_n$ uniformly chosen from $[n]^n$,
\[
\mathbb{E}[F_0(a)] = \sum_{\ell\in [n]}
\Pr_a[\exists i \in [n]\mbox{ such that } a_i= \ell] = n [ 1- (1-1/n)^n ].
\]
Hence $0.632n<(1-1/e)n<\mathbb{E}[F_0(a)]\le 0.75 n$.
Define a Doob martingale
$D_t$, $t= 0,1,\ldots, n$ with respect to the sequence $a_1\ldots a_n$
by $D_t = \mathbb{E}[F_0(a)\mid a_1\ldots a_t]$.
Therefore $D_0 = \mathbb{E}[F_0(a)]$ and $D_n = F_0(a)$.
Applying the Azuma-Hoeffding inequality, we have
\[
\Pr_a[F_0(a)\notin [0.5n,0.85n]]\le
\Pr_a[|F_0(a)-\mathbb{E}[F_0(a)| \geq 0.1n]
\le 2 e^{- 2\frac{(0.1n)^2}{n}} = 2 e^{-n/50},\]
which proves the claim.
\end{proof}

We say that $x_j$ is \emph{unique in $x$} iff 
$x_j \notin \{x_1,\ldots,x_{j-1},x_{j+1},\ldots,x_{2n-1}\}$.

\begin{lemma}\label{lemma:f0_range_prob}\label{manyruns}
Let $x$ be chosen uniformly at random from $[n]^{2n-1}$ with $n \geq 2$.
With probability at least $1-4n e^{-n/50}$,
\begin{enumerate}[(a)]
\item all outputs of $F_0^{\boxplus n}(x)$ are between $0.5n$ and $0.85n$,
and
\item the number of positions $j < n$ such that $x_j$ is unique in $x$ is at
least $n/24$.
\end{enumerate}
\end{lemma}

\begin{proof}
We know from Lemma~\ref{lemma:f0_range_prob_a} and the union bound that part (a) is false with probability at most $2 n e^{-n/50}$. 
For any $j<n$, let $U_j$ be the indicator variable of the event that $j$ is
unique in $x$ and $U=\sum_{j<n} U_j$.
Now $\E(U_j)=(1-1/n)^{2n-2}$ so $\E(U)= (n-1) (1-1/n)^{2n-2} \geq n/8$ for
$n\ge 2$.
Observe also that this is a kind of typical ``balls in bins" problem 
and so, as discussed in~\cite{dp:concentration-book}, it has the
property that the random variables $U_j$ are \emph{negatively
associated}; for example, for disjoint $A,A'\subset[n-1]$,  the larger 
$\sum_{j\in A} U_j$ is, the smaller $\sum_{j\in A'} U_j$ is likely to be.
Hence, it follows~\cite{dp:concentration-book}
that $U$ is more closely concentrated around its mean than if
the $U_j$ were fully independent.
It also therefore follows that we can apply a Chernoff bound directly to our
problem, giving
$\Pr[U \leq n/24] \leq \Pr[U \leq \E(U)/3] \leq e^{-2\E(U)/9} \leq e^{-n/36}$.
We obtain the desired bound for parts (a) and (b) together by another
application of the union bound.
\end{proof}

\paragraph{Correctness of a small branching program for computing outputs
in $\pi$-unique positions}

\begin{definition}
Let $B'$ be an $[n]$-way branching program and let $\pi$ be a source-sink
path in $B'$ with queries $Q_\pi$ and answers $A_\pi:Q_\pi\rightarrow [n]$.  
An index $\ell<n$ is said to be \emph{$\pi$-unique} iff 
either (a) $\ell \notin Q_\pi$, or 
(b) $A_\pi(\ell) \notin A_\pi(Q_\pi-\{\ell\})$.
\end{definition}

In order to measure the correctness of a small branching program, we 
restrict our attention to outputs that are produced at positions that are
$\pi$-unique and upper-bound the probability that a small branching program
correctly computes outputs of $F_0^{\boxplus n}$ at many $\pi$-unique positions
in the input.

Let $\mathcal{E}$ be the event that all outputs of $F_0^{\boxplus n}(x)$
are between $0.5n$ and $0.85n$.

\begin{lemma}
\label{runs-lemma}
Let $r> 0$ be a positive integer, let $\epsilon \le 1/10$, and let
$B'$ be an $[n]$-way branching program of height $q=\epsilon n$.
Let $\pi$ be a path in $B'$ on which
outputs from at least $r$ $\pi$-unique positions are produced.
For random $x$ uniformly chosen from $[n]^{2n-1}$,
\[
\Pr[\mbox{these $r$ outputs are correct for $F_0^{\boxplus n}(x)$},\mathcal{E}
\mid \pi_{B'}(x)=\pi]\le (17/18)^r.
\]
\end{lemma}

\begin{proof}
Roughly, we will show that when $\mathcal{E}$ holds (outputs for all
windows are not extreme) then, conditioned
on following any path $\pi$ in $B'$, each output
produced for a $\pi$-unique position will have only a constant probability
of success conditioned on any outcome for the previous outputs.  Because of the
way outputs are indexed, it will be convenient to consider these outputs
in right-to-left order.

Let $\pi$ be a path in $B'$, $Q_\pi$ be the set of queries along $\pi$,
$A_\pi:Q_\pi\rightarrow [n]$ be the answers along $\pi$, and
$Z_\pi:[n]\rightarrow [n]$ be the partial function denoting the outputs 
produced along $\pi$.
Note that $\pi_{B'}(x)=\pi$ if and only if $x_i=A_\pi(i)$ for all $i\in Q_\pi$.

Let $1\le i_1<i_2<\ldots < i_r<n$ be the first $r$ of the $\pi$-unique
positions on which $\pi$ produces output values; i.e.,
$\{i_1,\ldots, i_r\}\subseteq \mathrm{dom}(Z_\pi)$.
Define $z_{i_1}=Z_\pi(i_1),\ldots, z_{i_r}=Z_\pi(i_r)$.

We will decompose the probability over the input $x$ that $\mathcal{E}$ and all of $y_{i_1}=z_{i_1},\ldots, y_{i_r}=z_{i_r}$ hold via the chain rule.   
In order to do so,
for $\ell \in [r]$, we define event 
$\mathcal{E}_\ell$ to be $0.5n \le F_0^{(i)}(x)\le 0.85n$ for all $i>i_{\ell}$.
We also write $\mathcal{E}_{0} \myeq \mathcal{E}$.  Then
\begin{align}
&\Pr[y_{i_1}=z_{i_1},\ldots,y_{i_r}=z_{i_r},\ \mathcal{E}
\mid \pi_{B'}(x)=\pi]\nonumber\\
& = \Pr[\mathcal{E}_r\mid \pi_{B'}(x)=\pi]\cdot \prod_{\ell = 1}^{r} \Pr [y_{i_\ell}=z_{i_\ell},\ \mathcal{E}_{\ell-1}\mid
y_{i_{\ell+1}}= z_{i_{\ell+1}},\ \ldots,\ y_{i_{r}}= z_{i_{r}},\ \mathcal{E}_\ell,
\ \pi_{B'}(x)=\pi]\nonumber\\
& \le \prod_{\ell = 1}^{r} \Pr [y_{i_\ell}=z_{i_\ell}\mid
y_{i_{\ell+1}}= z_{i_{\ell+1}},\ \ldots,\ y_{i_{r}}= z_{i_{r}},\ \mathcal{E}_\ell,
\ \pi_{B'}(x)=\pi].
\label{chain rule}
\end{align}

We now upper bound each term in the product in~\eqref{chain rule}.
Depending on how much larger $i_{\ell+1}$ is than $i_\ell$, the conditioning
on the value of $y_{i_{\ell+1}}$ may imply a lot of information about
the value of $y_{i_\ell}$, but we will show that even if we reveal more
about the input, the value of $y_{i_\ell}$ will still have a constant amount
of uncertainty.

For $i\in [n]$, let $W_i$ denote the vector of input elements
$(x_i,\ldots,x_{i+n-1})$, and note that $y_i=F_0(W_i)$; we call $W_i$ the
$i^\mathrm{th}$ window of $x$.
The values $y_i$ for different windows may be closely related.
In particular, adjacent windows $W_i$ and $W_{i+1}$ have numbers of distinct
elements that can differ by at most 1 and this depends on whether the
extreme end-points of the two windows, $x_i$ and $x_{i+n}$, appear among
their common elements $C_i=\{x_{i+1},\ldots,x_{i+n-1}\}$.
More precisely,
\begin{equation} 
y_{i} - y_{i+1}=
\textbf{1}_{\{x_{i} \not \in C_i\}}
-
\textbf{1}_{\{x_{i+n} \not \in C_i\}}.
\label{indicators}
\end{equation}
In light of \eqref{indicators}, the basic idea of our argument is that,
because $i_\ell$ is $\pi$-unique and because of the conditioning on
$\mathcal{E}_\ell$, there will be enough uncertainty about
whether or not
$x_{i_\ell} \in C_{i_\ell}$  to show that the value of
$y_{i_\ell}$ is uncertain even if we reveal 
\begin{enumerate}
\item the value of the indicator
$\textbf{1}_{\{x_{i_{\ell}+n} \not \in C_{i_{\ell}}\}}$, and
\item the value of the output $y_{i_{\ell}+1}$.
\end{enumerate}

We now make this idea precise in bounding each term in the product
in~\eqref{chain rule}, using $\mathcal{G}_{\ell+1}$ to denote the event 
$\{y_{i_{\ell+1}}= z_{i_{\ell+1}},\ \ldots,\ y_{i_{r}}= z_{i_{r}}\}$.
\begin{align}
\Pr &[y_{i_\ell}=z_{i_\ell}\mid \mathcal{G}_{\ell+1},\ \mathcal{E}_\ell,
\ \pi_{B'}(x)=\pi]\notag\\
=&\sum_{m=1}^n \sum_{b \in \{0,1\}}
\Pr [y_{i_\ell}= z_{i_\ell}\mid
y_{i_{\ell} +1} = m, \textbf{1}_{\{x_{i_{\ell}+n} \not \in C_{i_{\ell}}\}} =b,\mathcal{G}_{\ell+1}, \mathcal{E}_\ell,\pi_{B'}(x)=\pi]\notag\\
& \qquad\times \Pr[y_{i_{\ell} +1} = m , \textbf{1}_{\{x_{i_{\ell}+n} \not \in C_{i_{\ell}}\}} =b\mid \mathcal{G}_{\ell+1}, \mathcal{E}_\ell,\pi_{B'}(x)=\pi]\notag\\
\leq & \max_{\substack{m \in [0.5n,0.85n]\\ b\in \{0,1\}}}
\Pr [y_{i_\ell}= z_{i_\ell}\mid
y_{i_{\ell} +1} = m, \textbf{1}_{\{x_{i_{\ell}+n} \not \in C_{i_{\ell}}\}} =b,
\mathcal{G}_{\ell+1}, \mathcal{E}_\ell,\pi_{B'}(x)=\pi]\notag\\
= & \max_{\substack{m \in [0.5n,0.85n]\\ b\in \{0,1\}}}
\Pr [ \textbf{1}_{\{x_{i_{\ell}} \not \in C_{i_{\ell}}\}}
= z_{i_\ell} -m+b\mid \notag\\[-3ex]
&\qquad\qquad\qquad\qquad y_{i_{\ell} +1} = m,\textbf{1}_{\{x_{i_{\ell}+n} \not \in C_{i_{\ell}}\}} =b,
\mathcal{G}_{\ell+1}, \mathcal{E}_\ell,\pi_{B'}(x)=\pi]
\label{correct-diff}
\end{align}
where the inequality follows because the conditioning
on $\mathcal{E}_\ell$ implies that
$y_{i_\ell +1}$ is between $0.5n$ and $0.85n$ and
the last equality follows because of the conditioning together with
\eqref{indicators} applied with $i=i_\ell$.  
Obviously, unless $z_{i_\ell}-m+b\in \{0,1\}$ the probability of the
corresponding in the maximum in \eqref{correct-diff} will be 0.
We will derive our bound by showing that given all the conditioning in
\eqref{correct-diff}, the probability of the event
$\{x_{i_{\ell}} \not \in C_{i_{\ell}}\}$ is between $2/5$ and $17/18$
and hence each term in the product in \eqref{chain rule} is at most $17/18$.

\paragraph{Membership of $x_{i_\ell}$ in $C_{i_\ell}$:} 
First note that the conditions $y_{i_\ell+1}=m$ and 
$\textbf{1}_{\{x_{i_{\ell}+n} \not \in C_{i_{\ell}}\}} =b$ together
imply that $C_{i_{\ell}}$ contains precisely $m-b$ distinct values.
We now use the fact that $i_\ell$ is $\pi$-unique and, hence,
either $i_\ell\notin Q_\pi$ or $A_\pi(i_\ell)\notin A_\pi(Q_\pi-\{i_\ell\})$.

First consider the case that $i_\ell\notin Q_\pi$.  
By definition, the events $y_{i_\ell+1}=m$, 
$\textbf{1}_{\{x_{i_{\ell}+n} \not \in C_{i_{\ell}}\}} =b$,
$\mathcal{E}_\ell$, and
$\mathcal{G}_{\ell+1}$ only depend on $x_i$ for $i>i_\ell$ and 
the conditioning on $\pi_{B'}(x)=\pi$ is only a property of
$x_i$ for $i\in Q_\pi$.   Therefore, under all the conditioning in
\eqref{correct-diff},
$x_{i_\ell}$ is still a uniformly random value in $[n]$.
Therefore the probability that 
$x_{i_\ell}\in C_{i_\ell}$ is precisely $(m-b)/n$ in this case.

Now assume that $i_\ell\in Q_\pi$.  In this case, the conditioning
on $\pi_{B'}(x)=\pi$ implies that $x_{i_\ell}=A_\pi(i_\ell)$ is fixed
and not in $A_\pi(Q_\pi -\{i_\ell\})$.
Again, from the conditioning we know that $C_{i_\ell}$ contains precisely
$m-b$ distinct values.  
Some of the elements that occur in $C_{i_\ell}$ may be inferred from the
conditioning -- for example, their values may have been queried along $\pi$ --
but we will show that there is significant uncertainty about whether any
of them equals $A_\pi(i_\ell)$.
In this case we will show that the uncertainty persists even if we reveal
(condition on) the locations of
all occurences of the elements $A_\pi(Q_\pi -\{i_\ell\})$ among the $x_i$
for $i>i_\ell$.

Other than the information revealed about the occurences of the elements
$A_\pi(Q_\pi -\{i_\ell\})$ among the $x_i$ for $i>i_\ell$,
the conditioning on the events $y_{i_\ell+1}=m$, 
$\textbf{1}_{\{x_{i_{\ell}+n} \not \in C_{i_{\ell}}\}} =b$,
$\mathcal{E}_\ell$, and
$\mathcal{G}_{\ell+1}$,
only biases the numbers of distinct elements and patterns of equality among
inputs $x_i$ for $i>i_\ell$.
Further the conditioning on $\pi_{B'(x)}=\pi$ does not reveal anything
more about the inputs in $C_{i_\ell}$ than is given by the occurences of
$A_\pi(Q_\pi -\{i_\ell\})$.
Let $\mathcal{A}$ be the event that all the conditioning is true.

Let $q'=|A_\pi(Q_\pi -\{i_\ell\})|\le q-1$
and let $q''\le q'$ be the number of distinct elements of
$A_\pi(Q_\pi -\{i_\ell\})$ that appear in $C_{i_\ell}$.
Therefore, since the input is uniformly chosen,  subject to the conditioning,
there are $m-b-q''$ distinct elements
of $C_{i_\ell}$ not among $A_\pi(Q_\pi -\{i_\ell\})$, and these distinct
elements are
uniformly chosen from among the elements $[n]-A_\pi(Q_\pi-\{i_\ell\})$.
Therefore, the probability that any of these $m-b-q''$ elements is
equal to $x_{i_\ell}=A_\pi(i_\ell)$ is precisely $(m-b-q'')/(n-q')$ in this
case.

It remains to analyze the extreme cases of 
the probabilities $(m-b)/n$ and $(m-b-q'')/(n-q')$ from the discussion above.
Since $q=\epsilon n$, $q''\le q'\le q-1$, and $b\in \{0,1\}$, we have the
probability
$\Pr[x_{i_\ell} \in C_{i_\ell}\mid \mathcal{A}] \leq \frac{m}{n - q+1} \leq
\frac{0.85n}{n- \epsilon n} \leq
\frac{0.85n}{n(1-\epsilon)} \leq  0.85/(1-\epsilon)\le 17/18$ since
$\epsilon\le 1/10$.
Similarly,
$\Pr[x_{i_\ell} \notin C_{i_\ell}\mid \mathcal{A}] 
< 1- \frac{m-q}{n}\leq 1- \frac{ 0.5 n- \epsilon n}{n} \leq 0.5+\epsilon\le 3/5$ since $\epsilon \le 1/10$.
Plugging in the larger of these upper bounds in~\eqref{chain rule},
we get:
\[
\Pr [z_{i_1}, \ldots, z_{i_r} \mbox{ are correct for }
F_0^{\boxplus n}(x) ,\ \mathcal{E}\mid \pi_{B'}(x)=\pi]
\leq (17/18)^{r},
\]
which proves the lemma.
\end{proof}

\paragraph{Putting the Pieces Together}
We now combine the above lemmas.
Suppose that $T S \le n^2/4800$ and let $q=n/10$.
We can assume without loss of generality that $S\ge \log_2 n$ since
we need $T\ge n$ to determine even a single answer.

Consider the fraction of inputs in $[n]^{2n-1}$ on which $B$ correctly computes
$F_0^{\boxplus n}$.
By Lemma~\ref{manyruns},
for input $x$ chosen uniformly from $[n]^{2n-1}$, the probability that
$\mathcal{E}$ holds and there are at least
$n/24$ positions $j<n$ such that $x_j$ is unique in $x$
is at least $1-4ne^{-n/50}$.
Therefore, in order to be correct on any such $x$, $B$ must correctly
produce outputs from at least $n/24$ outputs at positions $j<n$ such that 
$x_j$ is unique in $x$.

For every such input $x$, by our earlier outline, one of the $2^S$
$[n]$-way branching programs $B'$ of height $q$ contained in $B$ produces
correct output values for $F_0^{\boxplus n}(x)$  in at least
$r=(n/24)q/T\ge 20S$ positions $j<n$ such that $x_j$ is unique in $x$.

We now note that for any $B'$, if $\pi=\pi_{B'}(x)$ then the fact
that $x_j$ for $j<n$ is unique in $x$ implies that $j$ must be $\pi$-unique.
Therefore, for all but a $4ne^{-n/50}$ fraction of inputs $x$
on which $B$ is correct, $\mathcal{E}$ holds for $x$ and there is one of the
$\le 2^S$ branching programs
$B'$ in $B$ of height $q$ such that the path $\pi=\pi_{B'}(x)$ produces
at least $20S$ outputs at $\pi$-unique positions 
that are correct for $x$.

Consider a single such program $B'$. By Lemma~\ref{runs-lemma}
for any path $\pi$ in $B'$, the fraction of inputs
$x$ such that $\pi_{B'}(x)=\pi$ for which $20S$ of these outputs are
correct for $x$ and produced at $\pi$-unique positions, and
$\mathcal{E}$ holds for $x$ is at most $(17/18)^{20S}< 3^{-S}$.
By Proposition~\ref{manyruns}, this same bound applies to the fraction
of all inputs $x$ with $\pi_{B'}(x)=\pi$ for which $20S$ of these outputs
are correct from $x$ and produced at $\pi$-unique positions, and
$\mathcal{E}$ holds for $x$ is at most $(17/18)^{20S}< 3^{-S}$.

Since the inputs following different paths in $B'$ are disjoint,
the fraction of all inputs $x$ for which $\mathcal{E}$ holds
and which follow some path in $B'$ that yields at least
$20S$ correct answers from distinct runs of $x$ is less than $3^{-S}$.
Since there are at most $2^S$ such height $q$ branching programs,
one of which must produce $20S$ correct outputs from distinct runs of $x$
for every remaining input,
in total only a $2^S 3^{-S}=(2/3)^S$ fraction of all inputs have these
outputs correctly produced.

In particular this implies that $B$ is correct on at most a
$4ne^{-n/50}+(2/3)^S$
fraction of inputs.
For $n$ sufficiently large
this is smaller than $1-\eta$ for any $\eta<1-2^{-\delta S}$ for some
$\delta>0$, which contradicts our original assumption.  This completes the proof of Theorem~\ref{main_theorem}.
\end{proof}

\paragraph{Lower bound for $ (F_0\bmod 2)^{\boxplus n}$}
We describe how to modify the proof of Theorem~\ref{main_theorem} for
computing $F_0^{\boxplus n}$
to derive the same lower bound for computing $(F_0\bmod 2)^{\boxplus n}$.
The only difference is in the proof of Lemma~\ref{runs-lemma}.
In this case, each output $y_i$ is $F_0(W_i)\bmod 2$ rather than $F_0(W_i)$
and \eqref{indicators} is replaced by
\begin{equation} 
y_{i} = 
(y_{i+1}+ \textbf{1}_{\{x_{i} \not \in C_i\}}
-
\textbf{1}_{\{x_{i+n} \not \in C_i\}})\bmod 2.
\end{equation}
The extra information revealed (conditioned on) will be the same as in the
case for $F_0^{\boxplus n}$ but, because the meaning of $y_i$ has changed,
the notation $y_{i_\ell+1}=m$ is replaced by $F_0(W_{i_\ell+1})=m$, $y_{i_\ell+1}$ is then $m\bmod 2$, and the
upper bound in \eqref{correct-diff} is replaced by
\begin{align*}
\max_{\substack{m \in [0.5n,0.85n]\\ b\in \{0,1\}}}
\Pr [ \textbf{1}_{\{x_{i_{\ell}} \not \in C_{i_{\ell}}\}}
=&(z_{i_\ell} -m+b)\bmod 2\mid \notag\\[-3ex]
&F_0(W_{i_{\ell} +1}) = m,\textbf{1}_{\{x_{i_{\ell}+n} \not \in C_{i_{\ell}}\}} =b,
\mathcal{G}_{\ell+1}, \mathcal{E}_\ell,\pi_{B'}(x)=\pi]
\end{align*}
The uncertain event is exactly the same as before, namely whether or not
$x_{i_{\ell}} \in C_{i_{\ell}}$ and the conditioning is essentially exactly
the same, yielding an upper bound of $17/18$.
Therefore the analogue of Lemma~\ref{runs-lemma} also holds for
$(F_0\bmod 2)^{\boxplus n}$ and hence the time-space tradeoff of
$T\cdot S\in \Omega(n^2)$  follows as before.

\paragraph{Lower Bound for  $F_k^{\boxplus n}$, $k\geq 2$}
We describe how to modify the proof of Theorem~\ref{main_theorem} for
computing $F_0^{\boxplus n}$
to derive the same lower bound for computing $F_k^{\boxplus n}$ for $k\ge 2$.
Again, the only difference is in the proof of Lemma~\ref{runs-lemma}.
The main change from the case of $F_0^{\boxplus n}$
is that we need to replace \eqref{indicators} relating the values of consecutive
outputs.
For $k\ge 2$,
recalling that $f^{(i)}_j$ is the frequency of symbol $j$ in window $W_i$, we
now have
\begin{equation}
\label{indicators-k}
 y_{i} - y_{i+1} =\left[\left(f^{(i)}_{x_i}\right)^k - \left(f^{(i)}_{x_i} - 1 \right)^k \right] - \left[\left(f^{(i+1)}_{x_{i+n}} \right)^k - \left(f^{(i+1)}_{x_{i+n}} -1 \right)^k\right].
\end{equation}
We follow the same outline as in the case $k=0$ in order to bound the
probability that $y_{i_\ell}=z_{i_\ell}$
but we reveal the following information, which is somewhat more
than in the $k = 0$ case: 
\begin{enumerate}
\item $y_{i_\ell+1}$, the value of the output immediately after $y_{i_\ell}$,
\item $F_{0}(W_{i_\ell+1})$,
the number of distinct elements in $W_{i_\ell+1}$, and
\item $f^{(i_\ell+1)}_{x_{i_\ell+n}}$,
the frequency of $x_{i_\ell+n}$ in  $W_{i_\ell+1}$.
\end{enumerate}
For $M\in \mathbb{N}$, $m\in [n]$ and $1\le f\le m$,
define $\mathcal{C}_{M,m,f}$ be the event that
$y_{i_\ell+1}=M$, 
$F_{0}(W_{i_\ell+1})=m$, and
$f^{(i_\ell+1)}_{x_{i_\ell+n}}=f$.  Note that $\mathcal{C}_{M,m,f}$ only
depends on the values in $W_{i_\ell+1}$, as was the case for the information
revealed in the case $k=0$.
As before we can then upper bound the $\ell^\mathrm{th}$ term in the
product given in \eqref{chain rule} by
\begin{equation}
\label{upper-k}
\max_{\substack{m \in [0.5n,0.85n]\\ M\in \mathbb{N},f\in [m]}}
\Pr [y_{i_\ell}= z_{i_\ell}\mid
\mathcal{C}_{M,m,f},
\mathcal{G}_{\ell+1}, \mathcal{E}_\ell,\pi_{B'}(x)=\pi]
\end{equation}
Now, by \eqref{indicators-k}, given event $\mathcal{C}_{M,m,f}$, we have
$y_{i_\ell}=z_{i_\ell}$ if and only if
$z_{i_\ell} - M =\left[\left(f^{(i_\ell)}_{x_{i_\ell}}\right)^k - \left(f^{(i_\ell)}_{x_{i_\ell}} - 1 \right)^k \right] - \left[f^k - (f-1)^k\right]$, which we can express as
as a constraint on its only free parameter $f^{(i)}_{x_i}$,
$$\left(f^{(i_\ell)}_{x_{i_\ell}}\right)^k - \left(f^{(i_\ell)}_{x_{i_\ell}} - 1 \right)^k 
=z_{i_\ell} - M - f^k + (f-1)^k.$$  
Observe that this constraint can be satisfied for at most one positive
integer value of $f^{(i_\ell)}_{x_{i_\ell}}$ and that, by definition,
$f^{(i_\ell)}_{x_{i_\ell}}\ge 1$.
Note that $f^{(i_\ell)}_{x_{i_\ell}}=1$ if and only if
$x_{i_\ell}\notin C_{i_\ell}$, where $C_{i_\ell}$ is defined as in the case
$k=0$.
The probability that 
$f^{(i_\ell)}_{x_{i_\ell}}$ takes on a particular value is at most the larger
of the probability that $f^{(i_\ell)}_{x_{i_\ell}}=1$ or that
$f^{(i_\ell)}_{x_{i_\ell}}> 1$
and hence \eqref{upper-k} is at most
\begin{equation*}
\max_{\substack{m \in [0.5n,0.85n]\\ M\in \mathbb{N},f\in [m],c\in \{0,1\}}}
\Pr [ \textbf{1}_{\{x_{i_{\ell}} \not \in C_{i_{\ell}}\}}=c
\mid \mathcal{C}_{M,m,f},
\mathcal{G}_{\ell+1}, \mathcal{E}_\ell,\pi_{B'}(x)=\pi]
\end{equation*}
We now can apply similar reasoning to the $k=0$ case to argue that this
is at most $17/18$:  The only difference is that 
$\mathcal{C}_{M,m,f}$ replaces the conditions 
$y_{i_\ell+1}=F_0(W_{i_\ell+1})=m$ and
$\textbf{1}_{\{x_{i_{\ell}+n} \not \in C_{i_{\ell}}\}} =b$.
It is not hard to see that the same reasoning still applies with the new
condition.  The rest of the proof follows as before.

\subsection{A time-space efficient algorithm for $F_k^{\boxplus n}$}

We now show that the above time-space tradeoff lower bound is nearly
optimal even for restricted RAM models.

\begin{theorem}
There is a comparison-based deterministic RAM algorithm for computing
$F^{\boxplus n}_k$ for any fixed integer $k \geq 0$ with time-space tradeoff
$T\cdot S\in O(n^2\log^2{n})$ for all space bounds $S$ with $\log n\le S\le n$.
\end{theorem}

\begin{proof}
We denote the $i$-th output by $y_i=F_k(x_i,\ldots, x_{i+n-1})$.
We first compute $y_1$ using the
comparison-based time $O(n^2/S)$ sorting
algorithm of Pagter and Rauhe~\cite{pr:comparison-sorting}.
This algorithm produces the list of outputs in order by building a space $S$
data structure $D$ over the $n$ inputs and then repeatedly removing and
returning the index of the smallest element from that structure using a {\sc POP} operation.
We perform {\sc POP} operations on $D$ and keep track of the last index popped.
We also will maintain the index $i$ of the previous symbol seen as well as a
counter that tells us the number of times the symbol has been seen so far.
When a new index $j$ is popped, we compare the symbol at that index with the
symbol at the saved index.
If they are equal, the counter is incremented.
Otherwise, we save the new index $j$, update the running total for $F_k$
using the $k$-th power of the counter just computed, and then reset that
counter to 1.

Let $S'=S/\log_2{n}$.
We compute the remaining outputs in $n/S'$ groups of $S'$ outputs at a time.
In particular, suppose that we have already computed $y_i$.
We compute $y_{i+1},\ldots,y_{i+S'}$ as follows:

We first build a single binary search tree for both $x_i,\ldots, x_{i+S'-1}$
and for $x_{i+n},\ldots, x_{i+n+S'-1}$ and include a pointer $p(j)$ from each
index $j$ to the leaf node it is associated with.
We call the elements $x_i,\ldots, x_{i+S'-1}$ the old elements and add them
starting from $x_{i+S'-1}$.
While doing so we maintain a counter $c_j$ for each index $j\in [i,i+S'-1]$ of
the number of times that $x_j$ appears to its right in $x_i,\ldots,x_{i+S'-1}$.
We do the same for $x_{i+n},\ldots, x_{i+n+S'-1}$, which we call the
new elements, but starting from the left.
For both sets of symbols, we also add the list of indices where each element
occurs to the relevant leaf in the binary search tree.

We then scan the $n-S'$ elements $x_{i+S'},\ldots, x_{i+n-1}$ and maintain
a counter $C(\ell)$ at each leaf $\ell$ of each tree to record the number of
times that the element has appeared.

For $j\in [i,i+S'-1]$ we produce $y_{j+1}$ from $y_j$.
If $x_j = x_{j+n}$ then $y_{j+1} = y_j$.
Otherwise, we can use the number of times the old symbol $x_j$ and the
new symbol $x_{j+n}$ occur in the window $x_{j+1},\dots,x_{j+n-1}$ to
give us $y_{j+1}$.
To compute the number of times $x_j$ occurs in the window, we look at
the current head pointer in the new element list associated with leaf $p(j)$
of the binary search tree.
Repeatedly move that pointer to the right if the next position in the list
of that position is at most $n+j-1$.
Call the new head position index $\ell$.
The number of occurrences of $x_j$  in $x_{j+1}, \ldots, x_{S'}$ and
$x_{n+1}, \ldots, x_{n+j}$ is now $c_j+c_\ell$.
The head pointer never moves backwards and so the total number of pointer
moves will be bounded by the number of new elements.
We can similarly compute the number of times  $x_{j+n}$ occurs in the
window by looking at the current head pointer in the old element list
associated with $p(j+n)$ and moving the pointer to the left until it is
at position no less than $j+1$.
Call the new head position in the old element list $\ell'$.

Finally, for $k>0$ we can output $y_{j+1}$ by subtracting
$(1+c_j+c_\ell +C(p(j)))^k - (c_j+c_\ell +C(p(j))^k$ from $y_j$ and adding
$(1+c_{j+n}+c_{\ell'} +C(p(j+n)))^k - (c_{j+n}+c_{\ell'} +C(p(j+n))^k$.
When $k=0$ we compute $y_{j+1}$ by subtracting the value of the indicator
$\textbf{1}_{c_j+c_\ell +C(p(j)) =0}$ from $y_j$ and adding $\textbf{1}_{c_{j+n}+c_{\ell'} +C(p(j+n)) =0 }$.

The total storage required for the search trees and pointers is $O(S'\log n)$
which is $O(S)$.
The total time to compute $y_{i+1},\ldots, y_{i+S'}$ is dominated by the
$n-S'$ increments of counters using the binary search tree, which is
$O(n\log S')$ and hence $O(n\log S)$ time. This computation must be done $(n-1)/S'$
times for a total of $O(\frac{n^2\log S}{S'})$ time.
Since $S'=S/\log n$, the total time including that to compute $y_1$ is
$O(\frac{n^2 \log n\log S}{S})$ and hence $T\cdot S\in O(n^2\log^2 n)$.
\end{proof}

\section{Element Distinctness is easier than $F_0\bmod 2$}\label{sec:ed}

In this section we investigate the complexity of  $ED^{\boxplus n}$ and show
that it is strictly easier than $(F_0\bmod 2)^{\boxplus n}$.  
This fact is established by giving a particularly simple errorless algorithm for
$ED^{\boxplus n}$ which runs in linear time on average on random inputs
with alphabet $[n]$ and hence also beats our strong average-case lower bound
for $(F_0\bmod 2)^{\boxplus n}$.

We also give a deterministic reduction showing that the complexity of
$ED^{\boxplus n}$ is very similar to that of $ED$.  In particular,
$ED^{\boxplus n}$ can be computed
with at most an $O(\log^2 n)$ additive increase in the space and $O(\log^2 n)$
multiplicative increase in time more than required to compute a single
instance of $ED$.
This shows that any
deterministic or randomized algorithm for $ED$ that satisfies
$T\cdot S\in o(n^2/\log^3 n)$ would provide a worst-case separation between
the complexities of $ED^{\boxplus n}$ and $(F_0\bmod 2)^{\boxplus n}$.

\subsection{A fast average case algorithm for $ED^{\boxplus n}$ with
alphabet $[n]$}

We show a simple average case $0$-error
sliding-window algorithm for $ED^{\boxplus n}$.
When the input alphabet is chosen uniformly at random from $[n]$, the
algorithm runs in $O(n)$ time on average using $O(\log{n})$ bits of space.
By way of contrast, in Section~\ref{sec:fk} we proved an average case
time-space lower bound of $\overline T\cdot \overline S \in \Omega(n^{2})$ for
$(F_0\bmod 2)^{\boxplus n}$ under the same distribution.


The method we employ is as follows.  We start at the first window of length
$n$ of the input and perform a search for the first duplicate pair starting
at the right-hand end of the window and going to the left.
We check if a symbol at position $j$ is involved in a duplicate by simply
scanning all the symbols to the right of position $j$ within the window.
If the algorithm finds a duplicate in a suffix of length $x$, it shifts the
window to the right by $n-x+1$ and repeats the procedure from this point.
If it does not find a duplicate at all in the whole window, it simply moves
the window on by one and starts again.

In order to establish the running time of this simple method,
we will make use of the following birthday-problem-related facts which we prove in Appendix~\ref{sec-facts}.

\begin{lemma}\label{lemma:birthday}
Assume that we sample i.u.d.\@ with replacement from the range
$\{1\dots n\}$ with $n \geq 4$.   Let $X$ be a discrete random variable that represents the number of samples taken when the first duplicate is found.
Then
\begin{equation}
\Pr\left(X \geq n/2\right) \leq e^{-\frac{n}{16}}. \label{birthday:0}
\end{equation}
We also have that
\begin{equation}
\mathbb{E}(X^2) \leq 4n. \label{birthday:2}
\end{equation}
\end{lemma}

We can now show the running time of our average case algorithm for
$ED^{\boxplus n}$.

\begin{theorem}
Assume that the input is sampled  i.u.d.\@ with replacement from alphabet $[n]$.
$ED^{\boxplus n}$ can be solved in  $T \in O(n)$ time on average and space
$S \in O(\log{n})$ bits.
\end{theorem}

\begin{proof}
Let $U$ be a sequence of values sampled uniformly from $[n]$ with $n\geq 4$.
Let $M$ be the index of the first duplicate in $U$ found when scanning from
the right and let $X=n-M$.
Let $W(X)$ be the number of comparisons required to find $X$.
Using our naive duplicate finding method we have that $W(X) \leq X(X+1)/2$.
It also follows from inequality~\eqref{birthday:2} that $\mathbb{E}(W) \leq 4n$.

Let $R(n)$ be the total running time of our algorithm and note that
$R(n) \leq n^3/2$.
Furthermore the residual running time at any intermediate stage of the
algorithm is at most $R(n)$.

Let us consider the first window and let $M_1$ be the index of the first
duplicate from the right and let $X_1 = n - M_1$.
If $X_1 \geq n/2$, denote the residual running time by $R^{(1)}$.
We know from~\eqref{birthday:0} that
$\Pr(X_1 \geq n/2) \leq e^{-\frac{n}{16}}$.
If $X_1 < n/2$, shift the window to the right by $M_1+1$ and find $X_2$ for
this new window.
If $X_2 \geq n/2$, denote the residual running time by $R^{(2)}$.
We know that $\Pr(X_2 \geq n/2) \leq e^{-\frac{n}{16}}$.
If $X_1 < n/2 $ and $X_2 < n/2$ then the algorithm will terminate,
outputting `not all distinct' for every window.

The expected running time is then
\begin{equation*}\begin{split}
\mathbb{E}(R(n)) &= E\left(W(X_1)\right) + E\left(R^{(1)}\right)\Pr\left(X_1 \geq \frac{n}{2}\right)\\
&\quad+\Pr\left(X_1 < \frac{n}{2}\right)\Bigl[E\left(W(X_2) \middle\vert X_1 < \frac{n}{2}\right) + E\left(R^{(2)}\right) \Pr\left(X_2 \geq \frac{n}{2} \middle\vert X_1 < \frac{n}{2}\right)\Bigr]\\
&\leq 4n + \frac{n^3}{2} e^{-\frac{n}{16}} + 4n + \frac{n^3}{2} e^{-\frac{n}{16}} \in O(n)
\end{split}\end{equation*}

The inequality follows from the followings three observations.
We know trivially that $\Pr(X_1 < n/2) \leq 1$.
Second, the number of comparisons $W(X_2)$ does not increase if some of
the elements in a window are known to be unique. Third,
$\Pr(X_2 \geq n/2 \land X_1 < n/2) \leq \Pr(X_2 \geq n/2) \leq e^{-\frac{n}{16}}$.
\end{proof}

We note that similar results can be shown for inputs uniformly chosen from the
alphabet $[cn]$ for any constant $c$.

\subsection{Sliding windows do not significantly increase the complexity of
element distinctness}\label{sec:slidingspacesaving}

As preparation for the main results of this section,
we first give a deterministic reduction which shows how the answer to
an element distinctness problem allows one to reduce the input size of
sliding-window algorithms for computing $ED_n^{\boxplus m}$.

\begin{lemma}
\label{windowED-size-reduction}
Let $n>m>0$.
\begin{enumerate}[(a)]
\item If $ED_{n-m+1}(x_m,\ldots, x_n)=0$ then
$ED_n^{\boxplus m}(x_1,\ldots, x_{n+m-1})=0^m$.
\item If $ED_{n-m+1}(x_m,\ldots, x_n)=1$ then define
\begin{enumerate}[i.]
\item $i_L=\max\{j\in [m-1]\mid ED_{n-j+1}(x_j,\ldots, x_n)=0\}$ where
$i_L=0$ if the set is empty and
\item $i_R=\min\{j\in [m-1]\mid ED_{n-m+j}(x_m,\ldots, x_{n+j})=0\}$ where $i_R=m$
if the set is empty.
\end{enumerate}
Then
$$ED_n^{\boxplus m}(x_1,\ldots, x_{n+m-1})=0^{i_L} 1^{m-i_L}\ \land\ 1^{i_R} 0^{m-i_R}\ \land\ ED_{m-1}^{\boxplus m}(x_1,\ldots,x_{m-1},x_{n+1},\ldots, x_{n+m-1})$$
where each $\land$ represents bit-wise conjunction.
\end{enumerate}
\end{lemma}

\begin{proof}
The elements $M=(x_m,\ldots, x_n)$ appear in all $m$ of the windows so if this
sequence contains duplicated elements, so do all of the windows and hence the
output for all windows is $0$.  This implies part (a).

If $M$ does not contain any duplicates then any duplicate
in a window must involve at least one element from $L=(x_1,\ldots, x_{m-1})$ or
from $R=(x_{n+1},\ldots, x_{n+m-1})$.
If a window has value 0 because it contains an element of $L$ that also appears in $M$, it must also contain the rightmost such element of $L$ and hence any
window that is distinct must begin to the right of this rightmost such element
of $L$.
Similarly, if a window has value 0 because it contains an element of $R$ that
also appears in $M$, it must also contain the leftmost such element of $L$ and
hence any window that is distinct must end to the left of this leftmost such
element of $R$.
The only remaining duplicates that can occur in a window can only involve
elements of both $L$ and $R$.
In order, the $m$ windows contain the following sequences of elements of $L\cup R$:
$(x_1,\ldots, x_{m-1})$, $(x_2,\ldots, x_{m-1}, x_{n+1})$, $\ldots$, $(x_{m-1},x_{n+1},\ldots, x_{n+m-2})$, $(x_{n+1},\ldots,x_{n+m-1})$.
These are precisely the sequences for which $ED_{m-1}^{\boxplus m}(x_1,\ldots,x_{m-1},x_{n+1},\ldots,x_{n+m-1})$ determines distinctness.   Hence
part (b) follows.
\end{proof}

We use the above reduction in input size to show that any efficient algorithm
for element distinctness can be extended to solve element distinctness over
sliding windows at a small additional cost.

\begin{theorem}
\label{windowED}
If there is an algorithm $A$ that solve element distinctness, $ED$, using
time at most $T(n)$ and
space at most $S(n)$, where $T$ and $S$ are nondecreasing functions of $n$,
then there is an algorithm $A^*$ that solves the sliding-window
version of element distinctness, $ED_n^{\boxplus n}$, in time $T^*(n)$ that is
$O(T(n)\log^2 n)$ and space $S^*(n)$ that is $O(S(n)+\log^2 n)$.  Moreover,
if $T(n)$ is $O(n^\beta)$ for $\beta>1$, then $T^*(n)$ is $O(n^\beta\log n)$.

If $A$ is deterministic then so is $A^*$.   If $A$ is randomized with error
at most $\epsilon$ then $A^*$ is randomized with error $o(1/n)$.  Moreover,
if $A$ has the obvious 1-sided error (it only reports that inputs are not
distinct if it is certain of the fact) then the same property holds for $A^*$.
\end{theorem}

\begin{proof}
We first assume that $A$ is deterministic.
Algorithm $A^*$ will compute the $n$ outputs of $ED_n^{\boxplus n}$ in $n/m$
groups of $m$ using
the input size reduction method from Lemma~\ref{windowED-size-reduction}.
In particular, for each group $A^*$ will first call $A$ on the middle section
of input size $n-m+1$ and output $0^m$
if $A$ returns 0.   Otherwise, $A^*$ will do two binary searches involving
at most $2\log m$ calls to $A$ on inputs of size at most $n$ to compute $i_L$
and $i_R$ as defined in part (b) of that lemma.
Finally, in each group, $A^*$ will make
one recursive call to $A^*$ on a problem of size $m$.

It is easy to see that this yields a recurrence of the form
$$T^*(n)=(n/m)[cT(n) \log m + T^*(m)].$$
In particular, if we choose $m=n/2$ then we obtain
$T^*(n)\le 2 T^*(n/2)+ 2c T(n)\log n$.   If $T(n)$ is $O(n^\beta)$ for $\beta>1$
this solves to $T^*(n)=O(n^\beta\log n)$.  Otherwise, it is immediate from
the definition of $T(n)$ that $T(n)$ must be $\Omega(n)$ and hence the recursion
for $A^*$ has $O(\log n)$ levels and the total cost associated with each
of the levels of the recursion is $O(T(n)\log n)$.

Observe that the space for all the calls to $A$ can be re-used in the recursion.
Also note that the algorithm $A^*$ only needs to remember a constant number
of pointers for each level of recursion for a total cost of $O(\log^2 n)$
additional bits.

We now suppose that the algorithm $A$ is randomized with error at most
$\epsilon$.
For the recursion based on Lemma~\ref{windowED-size-reduction}, we use
algorithm $A$ and run it $C=O(\log n)$ times on input $(x_m,\ldots,x_n)$,
taking the majority of the answers to reduce the error to $o(1/n^2)$.
In case that no duplicate is found in these calls, we then apply the noisy
binary search method of Feige, Peleg, Raghavan, and
Upfal~\cite{fpru:noisy-decision} to
determine $i_L$ and $i_R$ with error at most $o(1/n^2)$ by using only
$C=O(\log n)$ calls to $A$.   (If the
original problem size is $n$ we will use the same fixed number $C=O(\log n)$
of calls to $A$ even at deeper levels of the recursion so that each subproblem
has error $o(1/n^2)$.)
There are only $O(n)$ subproblems so the final error is $o(1/n)$.
The rest of the run-time analysis is the same as in the deterministic case.

If $A$ has only has false positives (if it claims that the input is not
distinct then it is certain that there is a duplicate) then observe that
$A^*$ will only have false positives.
\end{proof}

\section{Order Statistics in Sliding Windows}\label{sec:order}

We first show that when order statistics are extreme, their complexity over
sliding windows does not significantly increase over that of a single
instance.

\begin{theorem}
\label{maxUBthm}
There is a deterministic comparison algorithm that computes
$MAX_n^{\boxplus n}$ (equivalently $MIN_n^{\boxplus n}$) using
time $T\in O(n\log n)$ and space $S\in O(\log n)$.
\end{theorem}

\begin{proof}
Given an input $x$ of length $2n-1$, we consider the window of $n$
elements starting at position $ \lceil\frac{n}{2} \rceil$ and ending
at position $n+\lceil\frac{n}{2}\rceil-1$ and find the largest element
in this window naively in time $n$ and space $O(\log n)$; call it $m$.
Assume without loss of generality that $m$ occurs between positions
$ \lceil\frac{n}{2} \rceil$ and $n$, that is, the left half of the
window we just considered. Now we slide the window of length $n$ to
the left one position at a time. At each turn we just need to look
at the new symbol that is added to the window and compare it to $m$.
If it is larger than $m$ then set this as the new maximum for that
window and continue.

We now have all outputs for all windows that start in positions 1 to
$\lceil\frac{n}{2} \rceil$. For the remaining outputs, we now run
our algorithm recursively on the remaining
$n+\lceil\frac{n}{2}\rceil$-long region of the input.  We only need to maintain the left and right endpoints of the current region.
At each level in the recursion, the number of outputs is halved and
each level takes $O(n)$ time. Hence, the overall time complexity is
$O(n\log n)$ and the space is $O(\log n)$.
\end{proof}

In contrast when an order statistic is near the middle, such as the median,
we can derive a significant separation in complexity between the
sliding-window and a single instance.
This follows by a simple reduction and known
time-space tradeoff lower bounds for sorting~\cite{bc:sorting,bea:sorting}.

\begin{theorem}
Let $P$ be a branching program computing $O_t^{\boxplus n}$ in time
$T$ and space $S$ on an input of size $2n-1$, for any $t \in [n]$.
Then $T\cdot S \in \Omega(t^2)$ and the same bound applies to expected time
for randomized algorithms.
\end{theorem}

\begin{proof}
We give lower bound for $O_t^{\boxplus n}$ for $t \in[n]$ by showing a
reduction from sorting.
Given a sequence $s$ of $t$ elements to sort taking values in
$\{2, \ldots, n-1\}$, we create a $2n-1$ length string as follows: the first
$n-t$ symbols take the same value of $n$, the last $n-1$ symbols take
the same value of 1 and we embed the $t$ elements to sort in the
remaining $t$ positions, in an arbitrary order.
For the first window, $O_t$ is the maximum of the sequence $s$. As
we slide the window, we replace a symbol from the left, which has
value $n$, by a symbol from the right, which has value 1. The
$t^{th}$ smallest element of window $i = 1, \ldots,t $ is the
$i^{th}$ largest element in the sequence $s$. Then the first $t$
outputs of $O_t^{\boxplus n}$ are the $t$ elements of the sequence
$s$ output in increasing order.
The lower bound follows from~\cite{bc:sorting,bea:sorting}.
As with the bounds in~Corollary~\ref{average-and-random}, the proof methods
in~\cite{bc:sorting,bea:sorting} also immediately extend to average case and
randomized complexity.
\end{proof}

%

For the special case $t = \lceil
\frac{n}{2}\rceil$ (median), we note that the best 
lower bound known for the single-input version of the median problem is $T
\in \Omega(n \log \log_S n)$ derived in \cite{chan:selection-journal} for
$S\in \omega(\log n)$ and this is tight for the expected time of errorless
randomized algorithms.

\section*{Acknowledgements}
The authors would like to thank Aram Harrow for a number of insightful discussions and helpful comments during the preparation of this paper.
\bibliographystyle{plain}
\bibliography{theory,extra}

\appendix
\newpage

\section{Proof of Lemma~\ref{lemma:birthday}}\label{sec-facts}

Assume that the input is chosen uniformly from an alphabet $\Sigma$ with $|\Sigma| \geq 4$.
We will first establish a basic result about the probability of finding the first
duplicate after at least $x$ samples.
Taking the random variable $X$ to be as in Lemma~\ref{lemma:birthday},
we show the following fact.

\begin{lemma}
$ \Pr(X \geq x) \leq e^{-\frac{(x-1)^2}{2|\Sigma|}}.$ \label{lemma:birthdayapp}
\end{lemma}

\begin{proof}
The proof relies on the fact that $1-x \leq e^{-x}$.
\begin{align*}
\Pr(X \geq x) &=    \prod_{i=1}^{x-1}\left(1-\frac{i}{|\Sigma|}\right)\\
        &\leq   \prod_{i=1}^{x-1} e^{-\frac{i}{|\Sigma|}}\\
        &\leq e^{-\frac{x^2}{4|\Sigma|}}
\end{align*}
\end{proof}

Inequality~\eqref{birthday:0} now follows by substituting $x= n/2$ and
$|\Sigma| = n$ into Lemma~\ref{lemma:birthdayapp} giving
\[
 \Pr\left(X \geq \frac{n}{2}\right) \leq e^{-\frac{n}{16}}.
\]
To prove inequality~\eqref{birthday:2}, recall that for non-negative valued
discrete random variables
\[
\mathbb{E}(X) = \sum_{x=1}^{\infty} \Pr(X \geq x).
\]
Observe that
\begin{align*}
 \mathbb{E}(X^2) &= \sum_{x=1}^{\infty} \Pr(X^2 \geq x)\\
    &= \sum_{x=1}^{\infty} \Pr(X \geq \sqrt{x})\\
    &\leq \sum_{x=1}^{\infty} e^{-\frac{(\sqrt{x})^2}{4|\Sigma|}}\\
    &\leq \int_{x=0}^{\infty} e^{-\frac{(\sqrt{x})^2}{4|\Sigma|}}\\
    &= 4n.
\end{align*}

\end{document}